\documentclass[11pt, letterpaper]{article}
\usepackage[utf8]{inputenc}
\usepackage[letterpaper,hmargin=1in,top=1in,bottom=1in]{geometry}
\usepackage{setspace}
\usepackage{amsfonts,amsthm,amsmath,mathtools}
\usepackage{dsfont,newtxtext,newtxmath,microtype}
\usepackage{multicol}
\usepackage{todonotes}
\usepackage{enumitem}
\usepackage{physics}
\usepackage{graphicx, quantikz, tikz,tikz-3dplot}
\usetikzlibrary{shapes.geometric}
\usetikzlibrary{positioning}
\usetikzlibrary{arrows}
\usetikzlibrary{decorations.pathmorphing}
\usetikzlibrary{calc, fit}
\usepackage{float}
\usepackage{subcaption}
\usepackage{xcolor}
\usepackage{authblk}

\definecolor{navyblue}{rgb}{0.0, 0.0, 0.5}
\definecolor{LightPink}{rgb}{0.858, 0.188, 0.478}

\usepackage[colorlinks=true, urlcolor=blue,citecolor=purple,anchorcolor=blue, linkcolor=LightPink]{hyperref}
\usepackage{cleveref}

\newtheorem{theorem}{Theorem}
\newtheorem{definition}[theorem]{Definition}
\newtheorem{corollary}[theorem]{Corollary}

\newtheorem{lemma}[theorem]{Lemma}

\newtheorem{algorithm}[theorem]{Algorithm}

\newcommand{\pfail}{p_{\text{fail}}}

\DeclareMathOperator{\poly}{poly}

\usepackage[backend=bibtex8,style=alphabetic,doi=false,isbn=false,url=false,maxbibnames=5,sorting=nty,sortcites=false]{biblatex}
\bibliography{bibliography}

\title{Polynomial-Time Classical Simulation of Noisy IQP Circuits with Constant Depth}

\date{}

\author[1,2]{\href{https://orcid.org/0000-0001-6365-8238}{Joel~Rajakumar}}
\author[1,2]{\href{https://orcid.org/0000-0002-6077-4898}{James~D.~Watson}}
\author[1,3]{\href{https://orcid.org/0000-0001-7458-4721}{Yi-Kai Liu}}

\affil[1]{\normalsize  Joint Center for Quantum Information \& Computer Science, 
	National Institute of Standards \& Technology 
	\authorcr and University of Maryland, 
	College Park }
\affil[2]{\normalsize Department of Computer Science and Institute for Advanced Computer Studies,
	\authorcr 
	University of Maryland, College Park}
\affil[3]{\normalsize Applied and Computational Mathematics Division, National Institute of Standards and Technology (NIST)}

\begin{document}
	
	{\begingroup
		\hypersetup{urlcolor=navyblue}
		\maketitle
		\endgroup}
	
	\begin{abstract}
		Sampling from the output distributions of quantum computations comprising only commuting gates, known as instantaneous quantum polynomial (IQP) computations, is believed to be intractable for classical computers, and hence this task has become a leading candidate for testing the capabilities of quantum devices.
		Here we demonstrate that for an arbitrary IQP circuit undergoing dephasing or depolarizing noise, whose depth is greater than a critical $O(1)$ threshold, the output distribution can be efficiently sampled by a classical computer. 
		Unlike other simulation algorithms for quantum supremacy tasks, we do not require assumptions on the circuit's architecture, on anti-concentration properties, nor do we require $\Omega(\log(n))$ circuit depth.
		We take advantage of the fact that IQP circuits have deep sections of diagonal gates, which allows the noise to build up predictably and induce a large-scale breakdown of entanglement within the circuit.
		Our results suggest that quantum supremacy experiments based on IQP circuits may be more susceptible to classical simulation than previously thought. Our results also imply that fault-tolerance within the IQP framework cannot be achieved past a critical circuit depth.
  Furthermore, we show that the critical depth threshold of our algorithm is tight, and below this threshold there are noisy IQP circuits which are hard to sample from.
  Thus we demonstrate that noisy IQP circuits exhibit a phase transition in the computational complexity of sampling, as circuit depth is increased.

	\end{abstract}
	
	\section{Introduction}
	
	The field of quantum computation is based on the practical expectation that quantum computers will deliver large speed-ups compared to classical computers on important tasks.
	However, in most cases, implementing these algorithms is only expected to be feasible on a fault-tolerant quantum computer which is beyond the capabilities of current quantum devices.
	With the aim of experimentally demonstrating a quantum speed-up without error correction, a range of computational tasks --- so called ``quantum supremacy'' experiments --- have been devised which are expected to be difficult for a classical computer to solve, but which are tractable for a quantum computer. 
	Notably these include random circuit sampling, IQP circuit sampling, boson sampling, and many others\footnote{\cite{bremner2011classical, aaronson2011computational, lund2014boson, bouland2017complexity, chabaud2017continuous, boixo2018characterizing, martinis2018quantum}. See \cite{hangleiter2020sampling, hangleiter2206computational} for a more comprehensive reviews.}.
	Many of these protocols have been implemented on a variety of hardware on a scale approaching or achieving quantum supremacy\footnote{\cite{arute2019quantum, zhong2020quantum, wu2021strong, zhu2022quantum, madsen2022quantum, bluvstein2023logical}}.
	
	On the flip side, it has been argued that many of these implementations fall short of demonstrating ``quantum supremacy'' by showing that the quantum computation performed is actually simulatable (or spoof-able) using currently available classical computers. 
	In the process, many improvements have been made to simulation techniques, both of quantum circuits generally and of the particular supremacy experiments \cite{barak2020spoofing, villalonga2020establishing, chabaud2021classical, napp2022efficient, Aharonov_2023, maslov2024fast}.
	These simulation techniques are not only of interest for their practical results or relation to quantum supremacy experiments, but also because many of them elucidate the mechanism through which quantum circuits may become easy or hard to simulate, thus giving us insight into the question of where the power of quantum computers comes from.

	One of the most popular suggestions for demonstrating quantum supremacy, which we will focus on in this work, is the task of sampling from  Instantaneous Quantum Polynomial (IQP) circuits.
	IQP circuits, first introduced in \cite{jordan2009permutational, shepherd2010quantum} are
	circuits in which the gates commute, and so there is no notion of a temporal order to which gates are applied first.
	Such circuits are known to be classically hard to exactly sample from assuming the polynomial hierarchy does not collapse \cite{bremner2011classical}, and hard to approximately sample from assuming complexity-theoretic conjectures \cite{bremner2016average}.
	Moreover, sampling IQP circuits can be shown to be equivalent to sampling from the output of certain time-evolved Hamiltonians \cite{gao2017quantum, bermejo2018architectures}.
	The popularity of IQP sampling task is partly due to the fact that IQP circuits do not require a universal gate set, and thus are easier to practically implement, and some such experiments have been implemented \cite{bluvstein2023logical}.

	A key problem in the practical implementation of all quantum supremacy tests --- using IQP circuits or otherwise --- is that often the device we have access to is noisy.
	Intuitively, noise destroys the quantum properties of quantum computers, rendering them more easily classically simulatable.
	\citeauthor{Bremner_2017} \cite{Bremner_2017} show that IQP circuits can be approximately sampled from, assuming anticoncentration of circuit output distributions and a single layer of bit-flip noise before measurement, by exploiting a Fourier path method. \citeauthor{gao2018efficient} \cite{gao2018efficient} and \citeauthor{Aharonov_2023} \cite{Aharonov_2023} extend these techniques to demonstrate that in the presence of noise, random circuit sampling becomes classically simulatable assuming anticoncentration.
	The anticoncentration requirement comes up in proofs of hardness and easiness, and it is a statement about the `flatness' of the output distribution on average, when the quantum circuit is chosen randomly from some ensemble (see \cite{hangleiter2206computational} for details on this property). For example, \cite{Aharonov_2023} show that, assuming random quantum circuits anticoncentrate, there is a classical algorithm that samples from an $\epsilon$-approximation to the random quantum circuit's output distribution, with probability $1-\delta$, and runtime $\poly(n,1/\epsilon,1/\delta)$.

	In this work, we demonstrate a polynomial-time classical simulation algorithm for sampling from the output distribution of any IQP circuit beyond an $O(1)$ depth, in the presence of dephasing or depolarizing noise interspersed between gates.
	Importantly, our algorithm requires no assumption about the output distribution or distribution from which the IQP circuit is selected, and instead applies to every IQP circuit constructed from gates involving $O(1)$ qubits. 
	
	Our methods exploit the fact that Pauli noise inflicted on an IQP circuit has the effect of removing entanglement locally. 
	This allows us to break down our circuit into small subcircuits, each of which is not connected to the other subcircuits by entangling gates.
	We thus only have to simulate these smaller subcircuits disjointly, while the gates between them can be simulated classically.
	To show this, we employ results from graph percolation theory and concentration of measure.

	The onset of classical simulatability at $\Omega(1)$ depth due to noise has been observed recently in a different setting of computing expectation values for optimization problems \cite{dsfranca}. 
	To our knowledge, our results are the first to demonstrate the onset of classical simulatability at $\Omega(1)$ depth due to noise in the setting of sampling from the output distribution (for circuits whose noiseless implementations are thought to be hard to sample from). 

    Our sampling algorithm is efficient when the circuit depth is larger than a critical threshold which scales as $O(p^{-1}\log(p^{-1}))$, where $p$ is the strength of the noise.
    We furthermore show that this bound is tight by demonstrating the existence of hard-to-sample noisy IQP circuits (up to complexity-theoretic assumptions) at depths of $\Theta(p^{-1}\log(p^{-1}))$. To construct these circuits, we use existing results on fault-tolerance in low-depth IQP circuits  \cite{fujii2016computational}\cite{Bremner_2017}, which we describe in \cref{Sec: Tightness}.
    This demonstrates an asymptotically sharp phase transition in the complexity of exact sampling from the output distribution of a noisy IQP circuit at constant depth.
	
	Our results also place limitations on error mitigation techniques to achieve quantum advantage with noisy IQP circuits, which we discuss in \cref{Sec: Fault-Tolerance} and \cref{Sec: Anticoncentration}. Broadly, we show that fault-tolerant implementations of IQP circuits with $O(1)$-local operations of super-constant depth must deviate from the IQP framework. Some ways to do this include intermediate measurement with feed-forward, introducing non-diagonal gates, or supplying fresh qubits. For example, \cite{paletta2023robust} describes an implementation of IQP circuits that shows robustness to noise by performing a single step of intermediate measurement with feed-forward. Similarly, \cite{bluvstein2023logical} implement logical IQP circuits augmented with CNOTs.

	\section{Preliminaries and Notation}
	IQP circuits involve the application of a unitary which is diagonal in the computational basis on qubits initialized in the $\ket{+}$ state, followed by measurement in the Hadamard basis. The diagonal unitary is canonically constructed with gate sets such as $\{T,CS\}$, $\{e^{it Z},e^{it ZZ}\}$, or $\{Z,CZ,CCZ\}$ \cite{bremner2016average}. 
	Our results hold generally for any diagonal gate set and connectivity as long as each diagonal gate acts on $O(1)$ qubits. In our notation, we represent an arbitrary diagonal gate using CPTP map $\mathcal{D}$, and associated diagonal unitary matrix $D$, where $\mathcal{D}(\rho) = D\rho D^{\dagger}$. Our results also apply to IQP circuits augmented with SWAP gates (e.g. for implementations on architectures with restricted connectivity), but we omit discussion of this case as it follows straightforwardly from our results.
	\begin{figure}[H]
		\centering
		\input{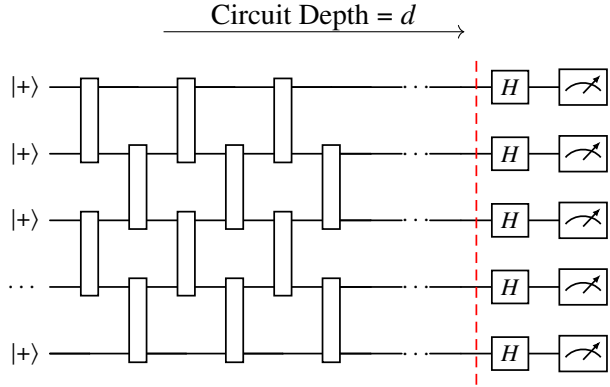}
		\caption{A generic IQP circuit on a 1D architecture involving $d$ layers of diagonal gates}
		\label{fig:IQP_Circuit}
	\end{figure}
	In this work, we consider noisy implementations of the diagonal portion of the IQP circuit, where independent Pauli noise channels are introduced on each qubit after each diagonal gate. Each interspersed Pauli noise channel can be expressed in the general form,
	\begin{align}
		\mathcal{N}_{p_X,p_Y,p_Z}(\rho) = p_I \rho + p_X X \rho X + p_Y Y \rho Y + p_Z Z\rho Z
	\end{align}
	where $p_I,p_X,p_Y,p_Z$ are non-negative probabilities that sum to 1, and it can be assumed that $p_I \geq 0.5$. We will often use $\mathcal{N}$ to refer to a Pauli noise channel of this form when the parameters are not relevant.

	A (potentially noisy) IQP circuit $C$ on qubits $1,\ldots,n$ is specified by a list of diagonal unitary channels (gates) and/or Pauli noise channels of known parameters that act on qubits in $1,\ldots,n$ in a fixed temporal order. Note we do not include Hadamard gates (or noise channels after the Hadamard gate) in the description of an IQP circuit --- instead, they are considered part of the `Hadamard-basis' measurement which occurs after the circuit $C$ is applied. 
	We use $\Phi_C$ to denote the CPTP map representing the action of (potentially noisy) circuit $C$ and $P_C$ to denote the output distribution over length-$n$ bitstrings obtained by Hadamard-basis measurement on the state $\Phi_C(\ketbra{+}^{\otimes n})$.
	That is, for a bitstring $b\in \{0,1\}^n$, we have:
	\begin{align*}
		\mathbf{P}_{X\sim P_{C} }(X=b) = \tr\left[ H^{\dagger \otimes n} \ketbra{b} H^{\otimes n} \Phi_C(\ketbra{+}^{\otimes n})\right].
	\end{align*}
	
	We also use subscripts on channels and density matrices to indicate which subsystems they apply to. In the case that the density matrix refers to a larger system, then we use the subscript to refer to the reduced density matrix on the substem indicated by the subscript (e.g. $\rho_A = \tr_A(\rho)$).

	\section{Classical Simulation of Noisy IQP circuits above Critical Depth}

	\begin{theorem} \label{Theorem:Main}
		Suppose $C$ is an IQP circuit containing $k$-local diagonal gates of depth $d$ on $n$ qubits. Let $\tilde{C}$ denote the noisy implementation of $C$, where each layer is interspersed with identical Pauli noise channels $\mathcal{N}_{p_X,p_Y,p_Z}$ on every qubit. Let $p = p_Z + \min(p_X,p_Y)$. 
		There exists a constant depth threshold $d_c \leq O(p^{-1}\log(kp^{-1}))$ such that when $d \geq d_c$, there exists a randomized classical algorithm that exactly samples from $P_{\tilde{C}}$ with random runtime $T$ of expected value,
		\begin{align*}
			\mathbf{E} [T] \leq O(dn^5).
		\end{align*}
	\end{theorem}
	
	\noindent Using standard techniques, we can convert the above `Las Vegas' algorithm into a `Monte Carlo' algorithm, which approximates the output distribution with a guaranteed worst-case runtime:
	\begin{corollary}
		Using the same notation as \cref{Theorem:Main}, there exists a constant depth threshold $d^* \leq O(p^{-1}\log(kp^{-1}))$, such that, when $d > d^*$,$p = \Omega(1)$ and $k=O(1)$, there exists a randomized classical algorithm that samples from $\tilde{Q}_{\tilde{C}}$, such that $\|\tilde{Q}_{\tilde{C}} - P_{\tilde{C}}\|_{TVD} \leq \epsilon$ for any $\epsilon > 0$, with worst-case runtime $T \leq O(d\poly(n/\epsilon))$.
	\end{corollary}
	
	Importantly, we note that the above algorithms allow us to sample efficiently for models of Pauli noise including \emph{dephasing} and \emph{depolarizing} noise. However, the dependency on $p = p_Z+\min(p_X,p_Y)$ implies our algorithm fails in the specific case that every noise channel is of the form $\mathcal{N}_{p_X,0,0}$ or of the form $\mathcal{N}_{0,p_Y,0}$ (when $p=0$).

	The constants hidden in the big-O notation are relatively small in practice, and the scaling of runtime with $n$ is a very loose upper bound. The main practical bottleneck for these algorithms is the scaling of depth thresholds $d^*$ and $d_c$ with $p$.
	We provide exact analytic expressions that relate $d^*$ and $d_c$ to noise parameters in the appendix, but as they are cumbersome to use, we plot them in \cref{fig:Comparison} for different values of $p$ for $k=2$. 
	We expect that similar analytic expressions for $d_c$ and $d^*$ can be found even if the noise channels are non-identical, but we omit discussion of this case for clarity. 
	\begin{figure}[h!]
		\centering
		\includegraphics[width=0.5\textwidth]{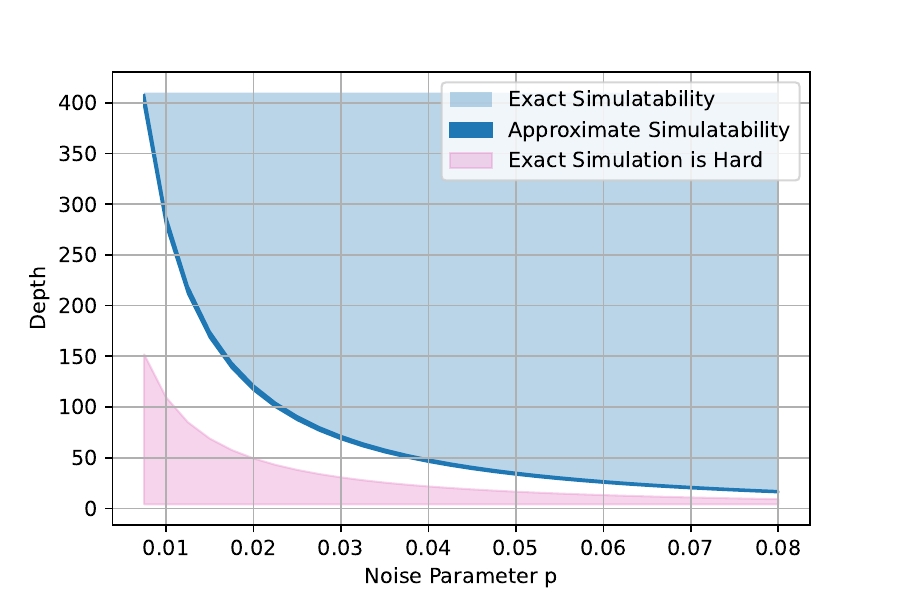}
		\caption{A plot of the depth thresholds for approximate sampling $(d^*)$ and exact sampling $(d_c)$. We also include the depth threshold for hardness of exact sampling, which we obtain in \cref{Sec: Tightness}}
		\label{fig:Comparison}
	\end{figure}

	\subsection{Motivation for the Classical Simulation Algorithm}
	\label{Sec:Motivation}
	
	The key idea behind the algorithm is that noise has the effect of removing entanglement that builds up in the circuit, and thus where noise appears, we can ``disentangle'' these parts of the circuit and simulate them classically.
	
	\begin{figure}[H]
		\centering
		\includegraphics[width=0.9\textwidth]{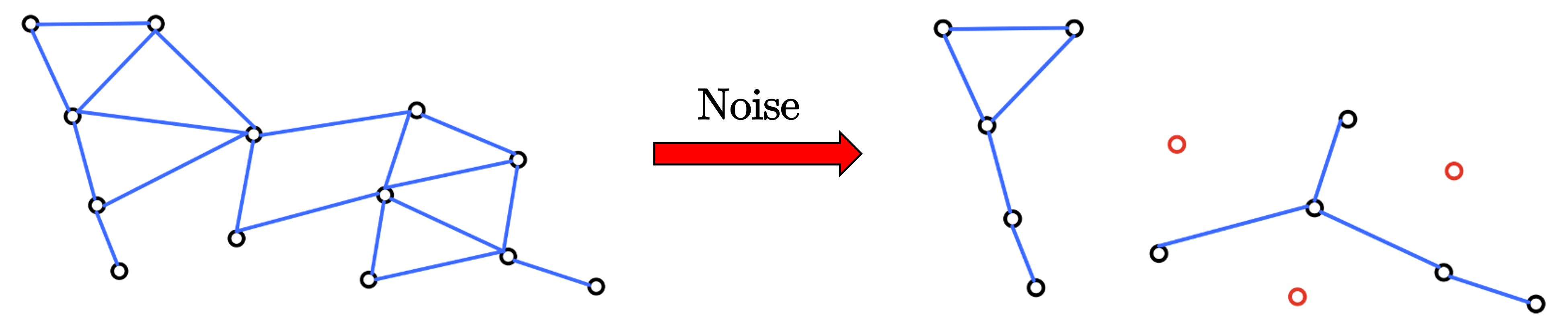}
		\caption{The graph on the left shows the interaction graph of a particular circuit, where the qubits are vertices and two vertices are joined by an edge if there is an entangling gate acting between them in the circuit. On the right, the red vertices indicate qubits which have been hit with noise. The effect of the noise is essentially to remove all interactions with this qubit from the larger circuit. The circuit can then effectively be simulated by considering only the connected components which are much smaller.}
		\label{fig:Interaction_Graph_with_Noise}
	\end{figure}
	
	In particular, for an IQP circuit $C$, we can define a corresponding interaction graph $G_C(V,E)$, where the qubits are vertices and two vertices are joined by an edge if there is an entangling gate acting between them in the circuit. 
	We will show that whenever a qubit receives an error, the edges of its corresponding vertex can be essentially ``removed'' from the interaction graph and its interactions become classically simulatable using a randomized classical algorithm. See \cref{fig:Interaction_Graph_with_Noise} as an example. 
	It is thus possible to simulate the noisy circuit by sampling a configuration of errors in the circuit, removing entangling gates acting on qubits that are affected by errors, and classically simulating each subcircuit corresponding to the remaining connected components \textit{independently} using standard circuit simulation techniques. 
	We will show these remaining subcircuits are sufficiently small that classically simulating them is efficient.

	\subsubsection{Sampling Completely Dephasing Errors to Reduce Entanglement}
	Our algorithm takes advantage of the fact that Pauli noise channels can be viewed as applying a completely dephasing error ($\mathcal{N}_{0,0,1/2}$ channel) with some fixed probability, along with correlated $X$-error or $Y$-error channels. This is shown in the following lemma which we prove in \cref{Sec:Noisy_Circuit_Proofs}.
	
	\begin{samepage}

	\begin{lemma} \label{Lemma:Noise}
		For any Pauli noise channel $\mathcal{N}_{p_X,p_Y,p_Z}$, define $p = p_Z + \min(p_X,p_Y)$, define $\mathcal{N}_1 = \mathcal{N}_{\frac{|p_X-p_Y|}{1-2p},0,0}$ if $p_X \geq p_Y$ or $\mathcal{N}_1 = \mathcal{N}_{0,\frac{|p_X-p_Y|}{1-2p},0}$ otherwise, and define $\mathcal{N}_2 = \mathcal{N}_{\frac{\min(p_X,p_Y)}{p},0,0}$. Then, for any single-qubit state $\rho$, 
		\begin{align*}
			\mathcal{N}_{p_X,p_Y,p_Z} (\rho) = (1-2p) \mathcal{N}_1(\rho)+ 2p \mathcal{N}_2 \circ \mathcal{N}_{0,0,1/2}(\rho).
		\end{align*}
	\end{lemma}
			
\end{samepage}
	\noindent The advantage of sampling completely dephasing errors is that they have very simple behavior in IQP circuits.
	In particular, one can show that the error can be commuted through to the beginning of the circuit, such that the effect of the error can be replicated by 
	replacing the initial $\ket{+}$ state of the qubit the error acts on with a random computational basis state. 
	This is summarized in the following lemma, proven in \cref{Sec:Noisy_Circuit_Proofs}.
	
	\begin{lemma} \label{Lemma:Commutation}
		Let $\tilde{C}$ be an IQP circuit with Pauli noise. 
		Let the initial state be $\ketbra{+}^n$. 
		Let $\tilde{C}'$ be the circuit $\tilde{C}$ where there is a \emph{single} completely dephasing error ($\mathcal{N}_{0,0,1/2}$ channel) on qubit $v \in \{1,\ldots n\}$ occurring \emph{at any point} in $\tilde{C}$. 
		Then,
		\begin{align}
			\Phi_{\tilde{C}'}(\ketbra{+}^n) = \mathbf{E}_{b \sim U(\{0,1\})}\left[\Phi_{\tilde{C}}(\ketbra{+}^{v-1} \otimes \ketbra{b} \otimes \ketbra{+}^{n-v})\right]
		\end{align}
	\end{lemma}
	
	\noindent The above lemma shows that dephasing errors force qubits to act classically in IQP circuits. 
	Next we show that diagonal gates in $\tilde{C}$ acting on computational basis states introduce no entanglement with the computational basis state.
	\begin{lemma} \label{Lemma:Diagonal}
		Let $A$ and $B$ be subsystems of qubits and $\mathcal{D}$ be any diagonal gate acting across these subsystems. Suppose subsystem $A$ is in computational basis state $\ketbra{b}$, and $\rho$ is the state of subsystem $B$. Define $D' = \tr_{A}((\ketbra{b}_A \otimes I_{B})D)$. Then,
		\begin{align}
			\mathcal{D}(\ketbra{b}_A \otimes \rho_B) = \ketbra{b}_A \otimes \mathcal{D'}(\rho_B)
		\end{align}
	\end{lemma}
	
	\noindent Thus, when a qubit receives a completely dephasing error in an IQP circuit $C$, the edges of its corresponding vertex in the interaction graph $G_{C}(V,E)$ are essentially ``removed" from the graph, because the circuit's diagonal gates no longer introduce entanglement with the decohered qubit. 
	The probability of a vertex being hit with noise and thus having its edges removed is $1-(1-2p)^d$ and is independent of noise on other vertices.

	\subsubsection{Phase Transition in Circuit Connectivity} \label{sec:Phase_Transition}
	We now want to examine the largest connected component of the IQP graph once edges are removed due to noise.
	The phenomenon of removing edges from randomly chosen vertices of a graph is well-studied in percolation theory as `vertex percolation,' and exactly corresponds to our setting. It is known that there is a phase transition after which all connected components are of $O(\log(n))$ size, as shown in the following lemma.
	\begin{lemma}(Informal) 
		Let $G = (V,E)$ be a graph of maximum degree $\Delta$  on n vertices. Construct $G' = (V,E')$ as follows. Start with $E' = E$. For each $v \in V$, with probability $1-q$, remove all edges incident to $v$ from $E'$. If $q < \frac{1}{\Delta }$, then whp all connected components of $G'$ are of size $O(\log n)$.
	\end{lemma}
	\noindent In our setting, if an IQP circuit $C$ of depth $d$ is constructed using $k$-local gates, then the max degree of the interaction graph $G_{C}$ is upper bounded by $\Delta \leq d(k-1)$ since a qubit can be acted on by at most $d$ gates, each of which entangles it with at most $k-1$ other qubits.
	The algorithm takes advantage of the fact that the probability of a qubit being included in a subcircuit decays inverse exponentially with depth, $q =(1-2p)^d$, while the probability required for percolation decays inverse linearly with depth, $q \leq 1/(k-1)d$. 
	Thus, there is a constant depth $d^*$ after which percolation into $O(\log(n))$-size subcircuits, which can be simulated tractably, occurs with high probability.

	\subsection{The Classical Simulation Algorithm} \label{Sec:Classical_Simulation_Algorithm}
	Here we give the classical simulation algorithm used to sample from a noisy IQC circuit.
	We leave the proof of runtime to \cref{Sec:Algorithm_Proof}, but prove correctness below in \cref{lemma:correctness}.

	\begin{samepage}
		\ \newline
		\noindent\makebox[\linewidth]{\rule{0.8\paperwidth}{0.4pt}}
		
		\begin{algorithm}[Sampler for Noisy IQP Circuits]
			\label{alg:1}
			
			Let $C,\tilde{C},\mathcal{N}_{p_x,p_y,p_z},p,n,d,k$ be defined as in \cref{Theorem:Main}. 
			Let $\mathcal{N}_1,\mathcal{N}_2$ be defined as in \cref{Lemma:Noise}.

			The algorithm stores and updates a classical description of the initial state of the IQP circuit as a list of $n$ characters $b = (b_1,\ldots,b_n)$ where each $b_i \in \{0,1,$`+'$\}$ for $i \in \{1,\ldots,n\}$ represents the initial state of qubit $i$ being one of $\{\ket{0},\ket{1},\ket{+}\}$ respectively. At the start of the algorithm, each $b_i =$ `+'. We use $\ket{b}$ to refer to the quantum state that $b$ describes, i.e. $\ket{b} = \bigotimes_i \ket{b_i}$. The algorithm outputs a list of $n$ measurement outcomes $s = (s_1,\ldots,s_n)$, where each $s_i \in \{0,1\}$ for $i \in \{1,\ldots,n\}$ represents the outcome of measurement on qubit $i$. The algorithm proceeds by modifying $\tilde{C}$ in stages $\tilde{C} \to \tilde{C}_1 \to \tilde{C}_2 \to \tilde{C}_3$, until $\Phi_{\tilde{C}_3}(\ketbra{b})$ can be simulated using state vector methods.
			\begin{enumerate} 
				\item \label{Alg1:Step_1}  
				Start with $\tilde{C}_1=\tilde{C}$. 
				For each single qubit channel $\mathcal{N}_{p_X,p_Y,p_Z}$, simulate $\mathcal{N}_{p_X,p_Y,p_Z}$  by replacing $\mathcal{N}_{p_X,p_Y,p_Z}\rightarrow \mathcal{N}_2 \circ \mathcal{N}_{0,0,1/2}$  with probability $2p$, and $\mathcal{N}_{p_X,p_Y,p_Z}\rightarrow \mathcal{N}_1$ otherwise.

				\item \label{Alg1:Step_2} 
				Start with $\tilde{C}_2 = \tilde{C}_1$. For each $i \in 1,\ldots,n$, if qubit $i$ receives a $\mathcal{N}_{0,0,1/2}$ error in $\tilde{C}_2$, update $b_i \sim U(\{0,1\}$. Remove all $\mathcal{N}_{0,0,1/2}$ channels from $\tilde{C}_2$.

				\item \label{Alg1:Step_3} 
				Start with $\tilde{C}_3$ empty. Iterate through each channel of $\tilde{C}_2$ in the temporal order in which it is applied and perform the following:
				\begin{enumerate}
					\item For each diagonal gate $\mathcal{D}$ on qubits $L \subseteq \{1,\ldots,n\}$, define $A = \{i: i \in L \wedge b_i \in \{0,1\}\}$, define $D' = \tr_{A}((\ketbra{b}_A \otimes I_{L - A})D)$, and add the equivalent implementation $\mathcal{D'}_{L - A}$ (which acts only on $L-A$ and leaves $\ketbra{b}_A$ unchanged) to $\tilde{C}_3$ .
					\item For each noise channel $\mathcal{N}_{p_X',p_Y',p_Z'}$ on qubit $i \in \{1,\ldots,n\}$, if $b_i\in \{0,1\}$, simulate this channel by updating $b_i \gets b_i \oplus 1$ with probability $p_X' + p_Y'$, and otherwise (if $b_i = $`+') add the channel to $\tilde{C}_3$  
					
				\end{enumerate}

				\item \label{Alg1:Step_4} Construct a graph $G_{\tilde{C}_3}(V,E)$ where $V = \{1,\ldots,n\}$ and $(v,w) \in E$ iff there is a diagonal gate acting between qubits $v$ and $w$ in $\tilde{C}_3$. Enumerate the connected components of $G_{\tilde{C}_3}$ as $V_1,\ldots,V_m$.
				
				\item \label{Alg1:Step_5}  Iterate through each $j \in 1,\ldots,m$. Suppose $V_j = \{v_1,\ldots,v_l\}$. If $|V_j|=1$ and $b_{v_1} \in \{0,1\}$, sample $s_{v_1} \sim U(\{0,1\})$. Otherwise, define the subcircuit $\tilde{C}_{3,j}$ to include only those channels in $\tilde{C}_{3,j}$ acting entirely on qubits of $V_j$, simulate the action of $\tilde{C}_{3,j}$ on qubits of $V_j$ with state vector methods, and sample $(s_{v_1},\ldots,s_{v_l}) \sim P_{\tilde{C}_{3,j}}$.
			\end{enumerate}
		\end{algorithm}
		\noindent\makebox[\linewidth]{\rule{0.8\paperwidth}{0.4pt}}
		
	\end{samepage}

	\begin{figure}[!h]
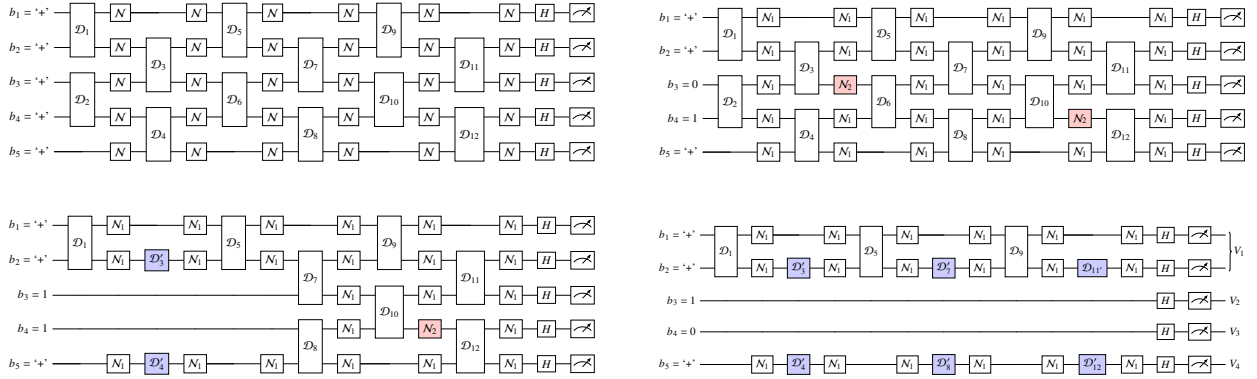

		\begin{subfigure}[b]{0.48\textwidth}
			\centering
			\resizebox{\linewidth}{!}{\input{tikz/Noisy_IQP_Circuit}}  
			\label{fig:A}
		\end{subfigure}\hspace{0.04\textwidth}
		\begin{subfigure}[b]{0.48\textwidth}
			\centering
			\resizebox{\linewidth}{!}{\input{tikz/Step1}}  
			\label{fig:B}
		\end{subfigure}\vspace{0.5cm}
		\begin{subfigure}[b]{0.48\textwidth}
			\centering
			\resizebox{\linewidth}{!}{\input{tikz/Step2}}  
			\label{fig:C}
		\end{subfigure}\hspace{0.04\textwidth}
		\begin{subfigure}[b]{0.48\textwidth}
			\centering
			\resizebox{\linewidth}{!}{\input{tikz/Step3}}  
			\label{fig:D}
		\end{subfigure}
		\caption{In (a), we show an example of $\tilde{C}$ on 5 qubits, and the list of characters $b = (b_1,b_2,b_3,b_4,b_5)$ which represents the initial state. In (b), we show the circuit $\tilde{C}_2$ constructed after steps 1 and 2. Certain noise channels, indicated in red, have been sampled as $\mathcal{N}_2 \circ \mathcal{N}_{0,0,1/2}$ (step 1), and after this, each $\mathcal{N}_{0,0,1/2}$ channel has been removed and its corresponding initial state in $b$ randomized, as $b_3=0,b_4=1$ (step 2). In (c), we show an intermediate layer of step 3, where the $\mathcal{D}'$ channels, indicated in blue, represent the replacements made in step 3a while $b_3 = 1$ due to step 3b. In (d), we show the final circuit $\tilde{C}_3$ and indicate the disjoint subsets $V_1,V_2,V_3,V_4$ into which it is partitioned for independent state vector simulation (steps 4 and 5).}
		\label{fig:Noise_on_Circuit}
	\end{figure}
	
	\begin{lemma} \label{lemma:correctness}
		For $\tilde{C}$ defined as in \cref{Theorem:Main}, let $Q_{\tilde{C}}(s)$ be the distribution over output strings $s$ produced by algorithm \ref{alg:1} on input $\tilde{C}$. Then, $Q_{\tilde{C}} = P_{\tilde{C}}$.
	\end{lemma}
	\begin{proof}
		We check the steps of the algorithm to ensure correct sampling:
		\begin{itemize}
			\item[Step \ref{Alg1:Step_1}:]  \Cref{Lemma:Noise} demonstrates that $\mathcal{N}_{p_x,p_y,p_z}$ can be implemented by probabilistically implementing $\mathcal{N}_2 \circ \mathcal{N}_{0,0,1/2}$ with probability $2p$ and $\mathcal{N}_1$ otherwise. Hence the channel we sample from remains the same, that is,
			\begin{align}
				\Phi_{\tilde{C}} = \mathbf{E}_{\tilde{C}_1} \left[\Phi_{\tilde{C}_1}\right]
			\end{align}
			\item[Step \ref{Alg1:Step_2}:] \Cref{Lemma:Commutation} shows that whenever a qubit is hit by a completely dephasing error ($\mathcal{N}_{0,0,1/2}$) in an IQP circuit, this can be simulated by replacing its initial state with a random computational basis state. 
            Let $b,b^{(0)}\in \{0,1, '+'\}^n,$ where $b^{(0)}$ denote the state of $b$ after step \ref{Alg1:Step_2} of the algorithm.
            Then,
			\begin{align}
				\Phi_{\tilde{C}_1}(\ketbra{+}^{\otimes n}) = \Phi_{\tilde{C}_2} \left( \mathbf{E}_{b^{(0)}} \left[\bigl| b^{(0)} \bigr\rangle \bigl\langle b^{(0)} \bigr|\right] \right)
			\end{align}
        
			\item[Step \ref{Alg1:Step_3}:] Observe the following,
			\begin{itemize}
				\item[(a) ] \Cref{Lemma:Diagonal} shows that if a set of qubits $A$ is in a computational basis state, then all diagonal gates acting on $A$ and another subsystem $L-A$ can be replaced by diagonal gates acting only on qubits in $L-A$: 
				\begin{align*} 
					\mathcal{D} (\ketbra{b}_A \otimes \rho_{L-A}) = \ketbra{b}_A \otimes \mathcal{D'} (\rho_{L-A})
				\end{align*}
                This leaves the state of the qubits in $A$ unchanged.
                
				\item[(b) ]  This step probabilistically implements noise channels on qubits in computational basis states. This works because $Y$ and $X$ errors act as bit-flip errors on $\ket{b}$ (because global phase introduced by $Y$ can be ignored) while $I$ and $Z$ errors act trivially. Thus, the channel $\mathcal{N}_{p_X',p_Y',p_Z'}$ acts on computational basis states by applying a bit-flip with probability $p_X'+p_Y'$.
			\end{itemize}
			Let $b^{(i)}$ denote the state of $b$ after the $i^{th}$ noise channel is encountered while iterating through the channels of $\tilde{C}_2$. Note that each randomly sampled `trajectory' $(b^{(0)},b^{(1)},\ldots,b^{(f)})$ of $b$ defines how the diagonal gates of $\tilde{C}_3$ are constructed from the diagonal gates of $\tilde{C}_2$, where we use $b^{(f)}$ to denote the final state of $b$. We have 
			\begin{align}
				\Phi_{\tilde{C}_2} \left( \bigl| b^{(0)} \bigr\rangle \bigl\langle b^{(0)} \bigr| \right) = \mathbf{E}_{b^{(0)},b^{(1)},\ldots,b^{(f)}} \Phi_{\tilde{C}_3} \left( \bigl| b^{(f)} \bigr\rangle \bigl\langle b^{(f)} \bigr| \right)
			\end{align}
			Moreover, regardless of the trajectory, $\tilde{C}_3$ contains no channels acting on qubits that are initialized in a computational basis state in $b^{(f)}$.
			\item[Step \ref{Alg1:Step_4}:] This step partitions the qubits into subsets $V_1,\ldots,V_m$ such that $\tilde{C}_3$ contains no diagonal gates crossing any partition.
			\item[Step \ref{Alg1:Step_5}:] Due to step \ref{Alg1:Step_4}, we can iterate through each $j = 1,\ldots,m$ and simulate the portion of $\tilde{C}_3$ acting on each $V_j$ independently. Qubits that are initialized in a computational basis state in $b_{(f)}$ correspond to isolated vertices in $G_{\tilde{C}_3}$, and their connected components are of the form $V_j = \{v\}$ (size 1). The outcome of Hadamard-basis measurement on these qubits is uniformly random. All other connected components involve qubits initialized in the $\ket{+}$ state, which are simulated exactly using state vector simulation.
		\end{itemize}
		
	\end{proof}

	\begin{figure}[!h]
		\begin{subfigure}[b]{0.65\textwidth}
			\centering
			\includegraphics[width=\linewidth]{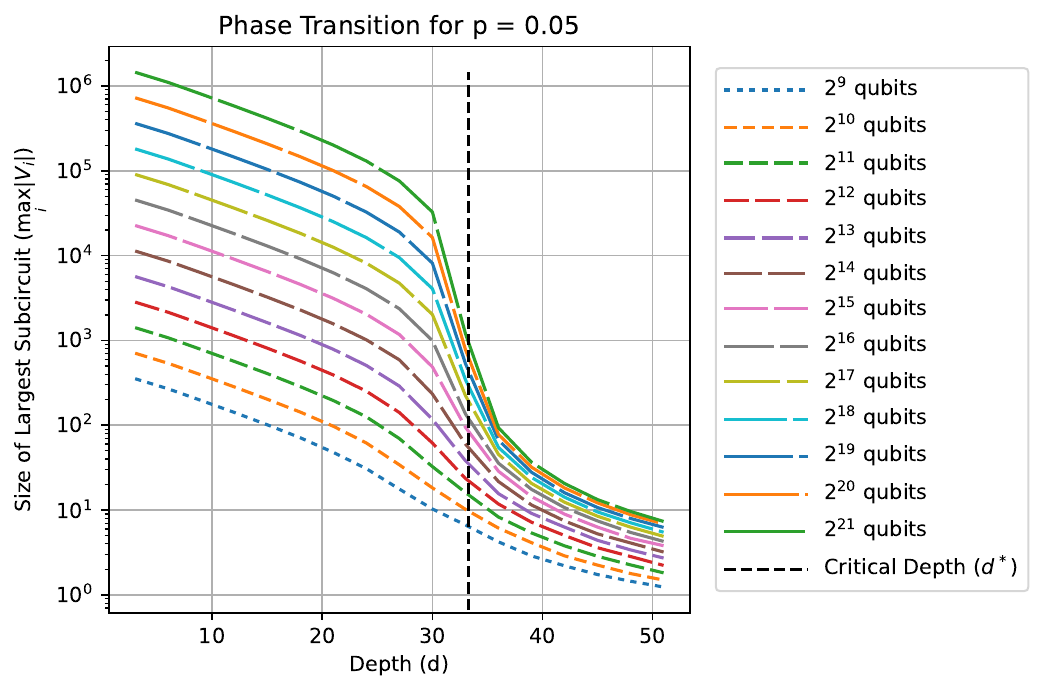}
		\end{subfigure}
		\begin{subfigure}[b]{0.35\textwidth}
			\centering
			\includegraphics[width=\linewidth]{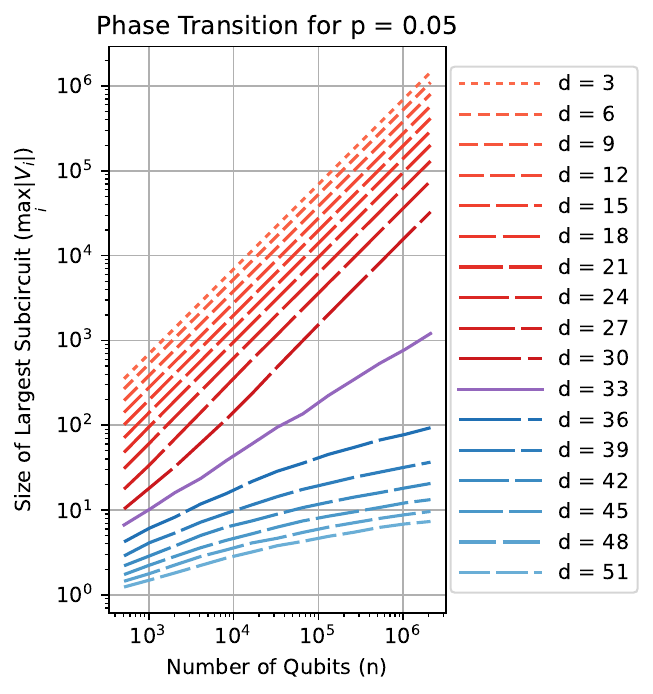}
		\end{subfigure}\vspace{0.5cm}
            \begin{subfigure}[b]{0.65\textwidth}
			\centering
			\includegraphics[width=\linewidth]{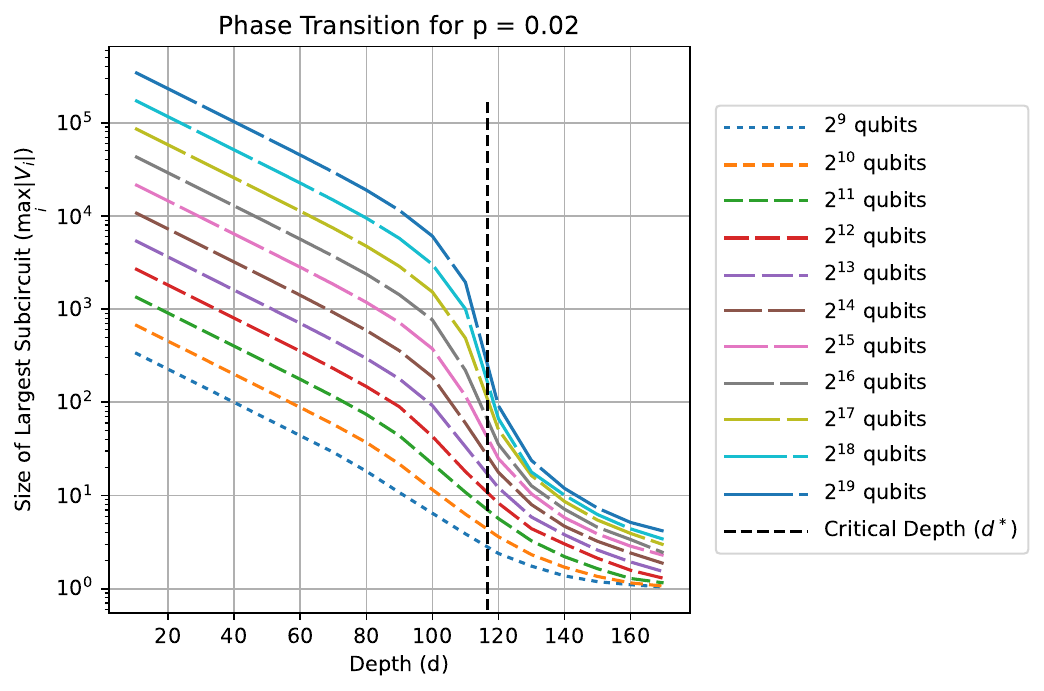}
		\end{subfigure}
		\begin{subfigure}[b]{0.35\textwidth}
			\centering
			\includegraphics[width=\linewidth]{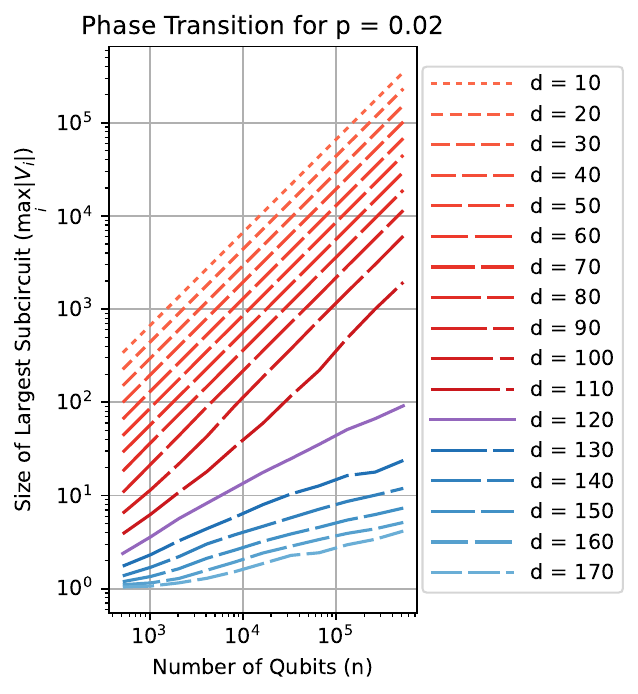}
		\end{subfigure}
		\caption{Using our analysis, the phase transition occurs at $d^* \approx 33$ for $p=0.05$ and $d^* = 117$ for $p=0.02$. In the figures on the left, we plot the size of the largest subcircuit ($\max_i |V_i|$) on a logarithmic scale against the depth of the circuit for different values of $n$. Before percolation occurs, the size of the largest subcircuit should decay exponentially with depth, because the number of non-decohered qubits decays as $\max_i |V_i| \leq (1-2p)^dn$: this is observed in the early linear portion of the plot. After the phase transition, the size should decay further to $\log(n)$ which is observed in the later portion of the plot (note that the y-axis values for each n, which were previously evenly spaced, get closer together). To directly observe the change in scaling with $n$, in the figures on the right, we plot the size of the largest circuit against the number of qubits on a log-log scale for different values of $d$. We see that depths below the phase transition show linear growth of the largest subcircuit with $n$ while depths above the phase transition show logarithmic growth of the largest subcircuit with $n$.}
		\label{fig:Numerics}
	\end{figure}
	
	\subsection{Numerical Observation of Phase Transition at Constant Depth}
	We observe the phase transition in classical simulatability by numerically studying the size of the largest subcircuit that needs to be simulated in \cref{alg:1} for various depths and numbers of qubits. We average over randomly constructed IQP circuits with interspersed noise of strength $p=0.05$ and $p=0.02$ and $2$-local gates (e.g. gate set $\{CS,T\}$). Because our algorithm does not depend on the randomness of the IQP circuits or the gate set used, we sample random $d$-regular graphs (which can always be implemented in either $d$ or $d+1$ layers by Vizing's theorem \cite{Vizing}) and observe the size of the largest connected component when vertices are kept with probability $(1-2p)^d$. The plots in \cref{fig:Numerics} show that a phase transition in the size of the largest subcircuit occurs at the same depth $d$, regardless of the number of qubits to start with. Moreover, this depth is observed to correspond to the analytic value of $d^*$, which shows our bound is reasonably tight.

 \section{Hardness of Noisy IQP circuits below Critical Depth} \label{Sec: Tightness}

So far we have demonstrated that noisy IQP circuits beyond an $O(p^{-1}\log p^{-1})$ depth can be efficiently sampled from.
It could be the case that this bound is not tight, and that there are other algorithms which are able to efficiently sample down to depths of $O(p^{-1})$, say.
Here we demonstrate that this is not the case, by constructing a class of IQP circuits with interspersed dephasing noise which are hard to exactly sample from, at depth $\Theta(p^{-1}\log p^{-1})$.
This provides a lower bound on the regime in which our sampling algorithm can work efficiently, and demonstrates that our bound on the critical depth for classical simulatability $d^* = O(p^{-1}\log p^{-1})$ is asymptotically sharp in terms of $p$.
This shows the existence of a sharp phase transition in the \textit{computational complexity} of noisy IQP circuits at constant depth.

We will construct an IQP circuit with interspersed dephasing noise. Note that this form of noise can be commuted to the end of the circuit, where it becomes bit-flip errors on the output distribution (due to hadamard-basis measurements). Therefore, we introduce the following notation.

 \begin{definition}[Noisy Circuit Sampling Distribution]
    For any noiseless IQP circuit $C$, we will use $P_C$ to denote the output distribution of sampled bitstrings from $C$, and $P_{C,q}$ to denote the output distribution of sampled bitstrings from $C$ with independent bit-flips applied on every output bit with probability $q$.
\end{definition}

Our proof makes use of three prior results, and we provide a brief exposition of each result in the appendix.
The first, from \cite{fujii2016computational}, is that there are noisy IQP circuits for which efficient exact sampling from the output collapses the polynomial hierarchy.
More formally:
\begin{lemma} (from \cite{fujii2016computational}) \label{hardness}
    There exists a family of IQP circuits $C$ constructed from a single layer of $e^{iZ\pi/8}$ gates and 4 layers of $e^{iZZ\pi/4}$ gates such that it is hard classically to exactly sample from the distribution $P_{C,\pfail} $ for $\pfail < 0.134$, unless the polynomial hierarchy collapses to the third level.
\end{lemma}

The high-level idea is to take this class of depth-5 circuits which are robust to low-level bit flip noise and construct noisy IQP circuits of depth $d$, which have identical output.
The fundamental problem here is that noise accumulates in deep circuits and will exceed the $0.134$ bound when $d > O(p^{-1})$, so we need some form of error mitigation.
To tackle this, we use a result from \cite{Bremner_2017} which shows that IQP circuits can be made fault-tolerant to arbitrary strengths of bit-flip noise on the output distribution, at the cost of non-locality in the gate set.  
The fundamental idea is that we can use a repetition code to repeat each qubit and gate $r$ times, and then decode with high probability.

\begin{lemma} (from \cite{Bremner_2017}) \label{encoding}
    Let $C$ be an arbitrary IQP circuit constructed of depth $d$ on $n$ qubits. Then, for any parameter $r$, there is an encoded IQP circuit $C'$ where every $k$-local gate in $C$ is encoded into a $kr$-local gate in $C'$, and a decoding algorithm $A$ such that,
    \begin{align}
        A(P_{C',q}) = P_{C, \pfail}
    \end{align}
    where $\pfail \leq (4q(1-q))^{r/2}$.
\end{lemma}

\noindent While \cref{encoding} results in an increase in the locality of the gates, we can reduce the locality of the gates to a constant at the expense of greater circuit depth using the following lemma, also referenced in \cite{Bremner_2017}.
\begin{lemma} (from \cite{shepherd2010binary} \label{locality}
    Let $C$ be an arbitrary IQP circuit constructed using gates of the form $e^{iZ^{\otimes k}\pi/8}$ for some $k$ a multiple of $3$, and gates of the form $e^{iZ^{\otimes k'}\pi/4}$ for some $k'$ a multiple of $2$. 
    There exists an equivalent circuit $C'$, such that $P_{C'} = P_C$, where every $e^{iZ^{\otimes k}\pi/8}$ gate in $C$ is decomposed into $k^2/2$ layers of $3$-local gates in $C'$, and every $e^{iZ^{\otimes k'}\pi/8}$ gate in $C$ is decomposed into $k'$ layers of $2$-local gates in $C'$.
\end{lemma}

We can then combine these results to show that there are deep circuits which survive the application of interspersed dephasing noise right up until the depth at which \cref{alg:1} becomes efficient.
\begin{theorem}[Deep \& Hard-to-Sample Circuits]
    There exist uniform families of noisy circuits of depth $d=\Theta(p^{-1}\log(p^{-1}))$ with interspersed dephasing noise of strength $p$, for which, if there exists an efficient, exact sampling algorithm for the output distribution, the polynomial hierarchy collapses to the third level. 
\end{theorem}
\begin{proof}
    From \cref{hardness}, there exists a family of IQP circuits $C$ such that if we can exactly sample from $P_{C,\pfail}$ when $\pfail < 0.134$, the polynomial hierarchy collapses to the third level.
    For any parameter $r$, applying \cref{encoding} to $C$ gives us a new family of IQP circuits $C'$ acting on $nr$ qubits which are robust to a higher level of output noise.
    $C'$ is a circuit consisting of 1 layer of $r$-local gates $e^{iZ^{\otimes r}\pi/8}$ and 4 layers of $2r$-local gates $e^{iZ^{\otimes 2r}\pi/4}$. 
    Choosing $r$ to be a multiple of $3$ and applying \cref{locality} to $C'$ gives us a new family of IQP circuits $C''$, involving only $3$-local gates on $nr$ qubits, such that $P_{C''} = P_{C'}$.
    We can bound the depth of this circuit as $d \leq r^2/2 + 4(2r) \leq r^2$ (for sufficiently large $r$).
    Finally, consider the circuit family $\tilde{C}''$, where $C''$ is interspersed with dephasing channels of strength $p$. These dephasing channels commute through the circuit and become bit-flip errors $N_{p_X,0,0}$ on the output, where $p_X = \frac{1-(1-2p)^d}{2}$. Now, using decoder $A$, we can relate $P_{\tilde{C''}}$ back to the hard-to-sample distribution $P_C$ as follows,
    \begin{align*}
    A(P_{\tilde{C}''}) &= A\left(P_{C'', \frac{1-(1-2p)^d}{2}}\right)\\
    &= A\left(P_{C', \frac{1-(1-2p)^d}{2}}\right)\\
    &= P_{C,\left( 4\left(\frac{1-(1-2p)^d}{2}\right)
        \left(\frac{1+(1-2p)^d}{2} \right)\right)^{\sqrt{d}/2}}
    \end{align*}
    As it is hard to sample from $P_{C,\pfail}$ for $\pfail < 0.134$, we can compute a sufficient condition for the hardness of exact sampling from $P_{\tilde{C}''}$ as,
    \begin{align*}
        \left( 4\left(\frac{1-(1-2p)^d}{2}\right)
        \left(\frac{1+(1-2p)^d}{2} \right)\right)^{\sqrt{d}/2} &< 0.134\\
        \left(1-(1-2p)^{2d}\right)^{\sqrt{d}/2} &< 0.134\\
         e^{-(1-2p)^{2d}\sqrt{d}/2} &< 0.134\\
         (1-2p)^{2d}\sqrt{d}/2 &> -\log(0.134)\\
         (1-2p)^{4d}4d &> 16(\log(0.134))^2
    \end{align*}
    Using \cref{Lemma:d} (in appendix), this inequality is satisfied when $d < \Theta(p^{-1}\log(p^{-1}))$. Thus, we can construct noisy IQP circuits of depth $d = \Theta(p^{-1}\log(p^{-1}))$ that are hard to exactly sample from.
\end{proof}

	\section{Applications \& Relation to Previous Work}

	\subsection{Limitations of Fault-Tolerance in IQP Circuits} \label{Sec: Fault-Tolerance}
	
	Since our results allow classical simulation of the ``worst-case" noisy circuit past some depth (e.g. the most noise-robust circuit possible at that depth), this allows us to rule out some forms of fault tolerance in IQP circuits in certain parameter regimes.
	Specifically, suppose there exists some fault-tolerance protocol $A$ that, for any noise parameter $p$, encodes a noiseless `logical' IQP circuit $C$ of $k$-local gates of depth $d$ on $n$ qubits into a `physical' circuit IQP $C'$ of $k$-local gates of depth $d' \geq d$ on $\poly(n)$ qubits with interspersed noise of strength $p$, such that one can recover an additive approximation to $P_C$ using samples from $P_{C'}$. 
	Our results show that $A$ cannot be defined for physical circuits of depth $d' \geq d^*$ (and therefore also for logical circuits of depth $d \geq d^*$), where $d^* \leq O(p^{-1}\log(kp^{-1}))$, assuming some complexity-theoretic conjectures. 
	
	The argument follows from \cite{bremner2016average} and \cite{Hangleiter_2018}, which show that there exist IQP circuits that are hard to sample from with additive approximation error, assuming some complexity-theoretic conjectures.
	If a protocol $A$ existed and worked correctly, we could take one of the encoded physical circuits, simulate it in polynomial time using \cref{alg:1}, and thus reproduce this `hard' distribution in polynomial time, thus causing complexity theoretic collapse. 
 
Intuitively, our algorithm suggests that fault-tolerance protocols which remain within the IQP framework cannot correct interspersed errors faster in depth than they build up (using only $O(1)$-local operations), and so must fail after some constant depth.

 \subsection{Anticoncentrated IQP Circuits and Classical Simulatibility} \label{Sec: Anticoncentration}
    It is interesting to consider how our classical simulation algorithm relates to previous work on demonstrating quantum computational advantage using IQP circuits that have anti-concentration behavior. In particular, \cite{Bremner_2017} provide a classical algorithm that approximately samples from noisy IQP circuits in quasi-polynomial time, assuming only anti-concentration of the output distribution (i.e. no requirements on circuit depth). It has been shown (\cite{Hangleiter_2018}) that certain families of IQP circuits anti-concentrate after only $4$ layers. Such circuits are leading candidates for quantum supremacy demonstrations, as this anti-concentration property is required in most proofs of hardness of approximate sampling \cite{bremner2016average,Hangleiter_2018,hangleiter2206computational}. However, for small noise strengths, a noisy implementation of the circuit family of \cite{Hangleiter_2018} would be classically simulatable using \cite{Bremner_2017}'s algorithm. In this regime, where one runs an anti-concentrated IQP circuit without any fault-tolerant encoding, our algorithm does not improve on the algorithm of \cite{Bremner_2017}. 
    
    In light of the classical simulation result of \cite{Bremner_2017}, a natural thought might be to encode the ``bare'' anti-concentrated IQP circuit into a noise-robust IQP circuit that is not anti-concentrated, and so cannot be simulated by the algorithm of \cite{Bremner_2017}. Indeed, this is exactly what \cite{Bremner_2017} propose in the later portion of their work. Our simulation algorithm rules out scalable quantum advantage with this idea, for any encoding scheme that remains within the IQP framework and requires a $\omega(1)$ depth overhead. For example, the encoding scheme described in \cite{Bremner_2017} results in an $\Omega((\log n)^2)$ blowup in depth, which would make the circuit simulatable using our techniques, or alternatively an $\Omega(\log n)$ blowup in gate locality, which would make it difficult to implement in practice.\footnote{This can be seen by applying \cref{encoding} and \cref{locality} with parameter $r \propto \log n$}

 \subsection{Computational Quantum-Classical Boundary of Noisy IQP Circuits}
	
	We compare our results to work by \citeauthor{fujii2016computational} \cite{fujii2016computational}, where they consider exact sampling from the output distribution of a family of fixed depth (e.g. depth-5) IQP circuits on $n$ with a variable noise strength $p$.
	They show that for this family, there exists a threshold noise strength $p_1$ such that for $p<p_1$, exact sampling from this family is classically hard (unless the polynomial hierarchy collapses to the third level) while there is a separate critical value $p_2$ such that if $p>p_2$, there is an efficient classical simulation algorithm for exact sampling.
 Our results extend their results, by providing an exact sampling algorithm for noisy IQP circuits at depths above some critical depth $d_1$, and a hardness result for noisy IQP circuits at depths below some critical depth $d_2$, where $d_1$ and $d_2$ are asymptotically equivalent functions of the noise strength $p$.
 Thus, we exhibit a similar ``Quantum-Classical Boundary'' to \citeauthor{fujii2016computational} \cite{fujii2016computational} in depth rather than noise strength.

	\subsection{Depth Lower Bounds for QAOA}
	The canonical QAOA circuit involves a tunable cost function Hamiltonian $H_c$, typically implemented with $R_{ZZ}(\gamma)$ gates, and a tunable mixer Hamiltonian $H_m$, typically this is the transverse-field mixer implemented with $R_{X}(\beta)$ gates on every qubit. These are alternatively applied in $r$ rounds on $\ket{+}$ states, followed by measurement in the computational basis. The final state is therefore
	\begin{align}
		\prod_{j=1}^r e^{iH_m\beta_j} e^{iH_c\gamma_j} \ket{+}^{\otimes n}
	\end{align}
	Our analysis applies to any QAOA circuit where the cost function is diagonal in the computational basis (i.e. the goal is to approximate the ground state of a classical Hamiltonian). As an example, we can consider QAOA for the MAXCUT problem.
	
	\begin{lemma}
        Let $G$ be a graph with maximum degree $\Delta$. The output distribution of a noisy $r$-round QAOA circuit to solve MAXCUT on $G$ can be sampled from with additive approximation error when $\Delta > \Delta^*$, where $\Delta^*$ is $O(p^{-1}\log(rp^{-1}))$.
	\end{lemma}
	\begin{proof}
		The key insight is that each round of QAOA resembles an IQP circuit. Suppose the depth $d$ required to implement $H_c$ is greater than the critical depth required for percolation to occur.
        Using the same sampling technique of Algorithm \ref{alg:1}, the resultant state vector can be classically stored as a set of decohered qubits, which are sampled as computational basis states, and a set of $O(\log n)$-sized connected components of non-decohered qubits. 
        
        The non-trivial issue to resolve is that subsequent rounds of QAOA can entangle the $O(\log(n))$-sized connected components with each other, making the state vector simulation unmanageable. That is, a decohered qubit (in a computational basis state) may be taken to some non-computational basis state by the application of $H_m$, after which the subsequent application of $H_c$ may entangle it with other qubits. One way to avoid this is for each decohered qubit to receive a dephasing error in every single round (between every $H_m$ application), thus keeping the decohered qubit in a computational basis state throughout the circuit. The probability that an error occurs on the same qubit in every round is $1-q = (1-(1-2p)^d)^r$. When $q < \frac{1}{\Delta}$, percolation will occur and the circuit will become classically simulatable. Because we are considering implementations of $H_C$ involving $R_{ZZ}$ gates, $\Delta \leq d$ (where $d$ is the depth of the circuit required to construct $H_C$). Therefore, we have,
		\begin{align*}
                q =
			(1-(1-(1-2p)^d)^r) \leq 
			(1-2p)^d\left[\sum_i^r (1-(1-2p)^d)^i\right] \leq 
			(1-2p)^dr \leq (1-2p)^{\Delta}r
		\end{align*}
        By \cref{Lemma:d} (in appendix), $(1-2p)^{\Delta}r < \frac{1}{\Delta}$ will be satisfied at $\Delta > \Delta^*$ for some critical degree $\Delta^* = O(p^{-1}\log(rp^{-1})$. At this point, percolation occurs with high probability, and therefore the circuit is classically simulatable.
	\end{proof}

	\section{Conclusion and Outlook}
	Our results show that dephasing and depolarizing noise make sampling from IQP circuits classically easy, thus putting tighter bounds on the regime in which quantum supremacy can be achieved using IQP sampling.
	
	On a theoretical level, our algorithm explicitly demonstrates how noise can remove quantum resources from a system to make it classically simulatable. In particular, our algorithm takes advantage of the fact that noise reduces the coherence of the initial state of the IQP circuit, and this, in turn, reduces the amount of entanglement that can be built up. 
    We speculate that the perspective of circuit percolation under noise may provide a framework for future classical simulation algorithms, both in IQP and more general cases, such as random quantum circuits. For example, recent work \cite{new_work} builds off our results to simulate noisy linear-optical circuits. We conjecture that the loss of quantum advantage at $\Omega(1)$ depths may be observed in a variety of NISQ circuits.
    In particular, we highlight the fact that most existing classical simulation algorithms for NISQ circuits take advantage of the convergence of the output distribution to the uniform distribution (due to noise or randomness of the circuit). This usually requires a depth that depends on the system size. However, our results are an example of convergence to a classically simulatable distribution that occurs at an $\Omega(1)$ depth which only depends on the noise strength. This is a common feature of graph percolation results --- that the onset of percolation is only determined by the local connectivity and the noise strength, but not the system size.

    We also note the similarities of our results with the phase transition in the cross-entropy benchmark for the task for noisy random circuit sampling observed in \cite{morvan2023phase, ware2023sharp}.
    Here the cross-entropy, which is used as a proxy for fidelity between the actual device and the idealized case, demonstrates that for a noisy quantum computer the device achieves good fidelity with the ideal quantum computer up to a sharp cut-off which depends on the circuit architecture.

    Finally, we note that our results here suggest that any fault-tolerant scheme for IQP circuits seems to require intermediate circuit measurements, or at least the ability to implement highly non-local gates. 
    In this sense, making an IQP circuit fault-tolerant (without leaving the IQP framework) seems to be qualitatively as hard as making a generic, non-IQP circuit fault-tolerant.

	\section*{Acknowledgements}
	{\begingroup
		\hypersetup{urlcolor=navyblue}
		We thank \href{https://orcid.org/0000-0002-4766-7967}{Dominik Hangleiter} and \href{https://orcid.org/0000-0003-3974-2987}{Michael Gullan}s for useful discussions about the sampling algorithm and IQP sampling.
        We also recognize useful comments and discussions with \href{https://orcid.org/0000-0001-5640-0343}{Ashley Montanaro} about this work and error correction in IQP circuits.
        Furthermore, we thank \href{https://orcid.org/0000-0002-8972-4413}{Alexander (Sasha) Barg} for discussions regarding graph percolation. We gratefully acknowledge discussions held at the Simons Institute workshop on ``Quantum Computational Advantage, Noisy Shallow Circuits, and Fault Tolerance," 2024. We are also thankful to anonymous reviewers for FOCS and SODA whose valuable comments helped improve our results.
		\endgroup}
	
	JR acknowledges support from the National Science Foundation Graduate Research Fellowship Program under Grant No. DGE 1840340.
	JDW acknowledges support from the United States Department of Energy, Office of Science, Office of Advanced Scientific Computing Research, Accelerated Research in Quantum Computing program, and also NSF QLCI grant OMA-2120757.
	YKL acknowledges support from NIST, DOE ASCR (Fundamental Algorithmic Research for Quantum Computing, award No. DE-SC0020312), NSF QLCI grant OMA-2120757, and ARO (Quantum Algorithms for Algebra and Discrete Optimization, grant number W911NF-20-1-0015).

	{\begingroup
		\hypersetup{urlcolor=navyblue}
		\printbibliography[heading=bibintoc]
		\endgroup}

	\appendix

	\section{Proofs of Noisy Circuit Lemmas}
	\label{Sec:Noisy_Circuit_Proofs}
	
	\begin{lemma}[Restatement of \cref{Lemma:Noise}] \label{Lemma:Noise_Appendix}
		For any Pauli noise channel $\mathcal{N}_{p_X,p_Y,p_Z}$, define $p = p_Z + \min(p_X,p_Y)$, define $\mathcal{N}_1 = \mathcal{N}_{\frac{|p_X-p_Y|}{1-2p},0,0}$ if $p_X \geq p_Y$ or $\mathcal{N}_1 = \mathcal{N}_{0,\frac{|p_X-p_Y|}{1-2p},0}$ otherwise, and define $\mathcal{N}_2 = \mathcal{N}_{\frac{\min(p_X,p_Y)}{p},0,0}$. Then, for any single-qubit state $\rho$, 
		\begin{align*}
			\mathcal{N}_{p_X,p_Y,p_Z} (\rho) = (1-2p) \mathcal{N}_1(\rho)+ 2p \mathcal{N}_2 \circ \mathcal{N}_{0,0,1/2}(\rho).
		\end{align*}
	\end{lemma}
	
	\begin{proof}
		Suppose $p_X \geq p_Y$. Then, $p = p_Z + p_Y$
		\begin{align*}
			\mathcal{N}_{p_X,p_Y,p_Z}(\rho) &= p_I \rho + p_X X \rho X + p_Y Y \rho Y  + p_Z Z \rho Z\\
			&= (p_I-p_Z) \rho + p_Z (\rho + Z\rho Z) +  (p_X - p_Y) X \rho X+ p_Y(X\rho X + Y\rho Y)  \\
			&= (p_I-p_Z) \rho + 2p_Z \mathcal{N}_{0,0,1/2}(\rho) +  (p_X - p_Y) X \rho X+ 2p_Y X\mathcal{N}_{0,0,1/2}(\rho )X \\
			&= (p_I-p_Z+p_X-p_Y) \mathcal{N}_{\frac{p_X-p_Y}{(p_I-p_Z+p_X-p_Y)}, 0 , 0}(\rho)+ 2 (p_Y +p_Z)  \mathcal{N}_{\frac{p_Y}{(p_Y+p_Z)},0,0} \circ \mathcal{N}_{0,0,1/2}(\rho)\\
			&= (1-2p) \mathcal{N}_{\frac{p_X-p_Y}{1-2p},0,0}(\rho)+ 2p  \mathcal{N}_{\frac{p_Y}{p},0,0}  \circ \mathcal{N}_{0,0,1/2}(\rho)
		\end{align*}
		where we have used that $p_I + p_X + p_Y+p_Z = 1$.
		The case when $p_Y > p_X$ follows using the same proof as above where we flip the role of the $X$ and $Y$ operators.
	\end{proof}
	
	\begin{lemma} [Restatement of \cref{Lemma:Commutation}] \label{Lemma:Commutation_Appendix}
		Let $\tilde{C}$ be an IQP circuit with Pauli noise. 
		Let the initial state be $\ketbra{+}^n$. 
		Let $\tilde{C}'$ be the circuit $\tilde{C}$ where there is a \emph{single} completely dephasing error ($\mathcal{N}_{0,0,1/2}$ channel) on qubit $v \in \{1,\ldots n\}$ occurring \emph{at any point} in $\tilde{C}$. 
		Then,
		\begin{align*}
			\Phi_{\tilde{C}'}(\ketbra{+}^n) = \mathbf{E}_{b \sim U(\{0,1\})}\left[\Phi_{\tilde{C}}(\ketbra{+}^{v-1} \otimes \ketbra{b} \otimes \ketbra{+}^{n-v+1})\right]
		\end{align*}
	\end{lemma}
	\begin{proof}
		The completely dephasing error trivially commutes with diagonal gates in the circuit, and commutes with all Pauli noise channels in the circuit because Pauli channels commute with each other (i.e. $\sigma_i \sigma_j \rho \sigma_j^\dagger \sigma_i^\dagger = \sigma_j \sigma_i \rho \sigma_i^\dagger \sigma_j^\dagger$, for any Pauli matrices $\sigma_i,\sigma_j$ and density matrix $\rho$). Thus, it can be commuted to the beginning of the circuit where it acts on the initial state as follows,
		\begin{align*}
			\mathcal{N}_{0,0,1/2}(\ketbra{+}) = \frac{1}{2}\ketbra{+} + \frac{1}{2}Z\ketbra{+}Z^\dagger = \frac{I}{2} = \mathbf{E}_{b} [\ketbra{b}].
		\end{align*}
	\end{proof}
	
	\begin{lemma} [Restatement of \cref{Lemma:Diagonal}] \label{Lemma:Diagonal_Appendix}
		Let $A$ and $B$ be subsystems of qubits and $\mathcal{D}$ be any diagonal gate acting across these subsystems. Suppose subsystem $A$ is in computational basis state $\ketbra{b}$, and $\rho$ is the state of subsystem $B$. Define $D' = \tr_{A}((\ketbra{b}_A \otimes I_{B})D)$. Then,
		\begin{align}
			\mathcal{D}(\ketbra{b}_A \otimes \rho_B) = \ketbra{b}_A \otimes \mathcal{D'}(\rho_B)
		\end{align}
	\end{lemma}
	\begin{proof}
		Any diagonal unitary matrix acting on subsystems $A$ and $B$ can be written in the following form
		\begin{align}
			D &= \sum_{i \in \{0,1\}^{|A|},j \in \{0,1\}^{|B|}} e^{i\theta_{ij}}\ketbra{ij}
		\end{align}
		where $\theta_{ij}$ are real numbers (phases). Therefore we have
		\begin{align*}
			\mathcal{D} \left(\ketbra{b}_A \otimes \rho_B\right) &= \sum_{i,i' \in \{0,1\}^{|A|},j,j' \in \{0,1\}^{|B|}} e^{i\theta_{ij}-i\theta_{i'j'}}\ketbra{ij} \left(\ketbra{b}_A\otimes \rho_B\right) \ketbra{i'j''}\\
			&= \ketbra{b}_A \otimes \sum_{j,j' \in \{0,1\}^{|B|}} e^{i\theta_{bj}-i\theta_{b j'}} \ketbra{j} \rho_B  \ketbra{j'}
		\end{align*}
		The lemma follows by observing that
		\begin{align*}
			\tr_A((\ketbra{b}_A \otimes I_{B})D) &= \tr_A\left(\sum_{j \in \{0,1\}^{|B|}} e^{i\theta_{bj}}\ketbra{bj}  \right)\\
			&= \sum_{j \in \{0,1\}^{|B|}} e^{i\theta_{bj}} \ketbra{j}
		\end{align*}
		and hence
		\begin{align*}
			\mathcal{D}'(\rho_B) = \sum_{j,j' \in \{0,1\}^{|B|}} e^{i\theta_{bj}-\theta_{bj'}} \ketbra{j} \rho_B \ketbra{j'}.
		\end{align*}
	\end{proof}

	\section{Proof of Percolation Bound} \label{Section:Percolation_Proof}
	It is folklore that graphs exhibit a phase transition in connectivity (i.e. the graph splits into smaller connected components) when elements of the graph are kept with probability which is low relative to the graph degree (see e.g. \cite{Grimmett1999}). \cite{krivelevich2015phase} explicitly proves bounds on the size of components in this regime, and we adapt these results to our setting.

	\begin{lemma}\label{Lemma:Percolation}
		Let $G = (V,E)$ be a graph of maximum degree $\Delta$  on n vertices. Construct $G' = (V,E')$ as follows. Start with $E' = E$. For each $v \in V$, with probability $1-q$, remove all edges incident to $v$ from $E'$. Let $V_1,\ldots, V_m$ be the connected components of $G'$.
		If $q < \frac{1}{\Delta }$, then for each $i \in 1,\ldots,m$,
		\begin{align}
			P(|V_i| >x) \leq e^{-x(1 -q\Delta  - \log (q\Delta ))}
		\end{align}
	\end{lemma}

	\begin{proof} 
		We describe a randomized algorithm (inspired from \cite{krivelevich2015phase}) that constructs a random graph $G'$ that is sampled from the probability distribution described in the statement of the lemma. The algorithm initializes $G'$ to be the empty graph. Then it constructs a set $S$ through probabilistic `queries,' which we define as follows. A query to vertex $v \in V$ ‘succeeds’ with probability $q$, in which case $v$ is added to $S$, or ‘fails’ otherwise, in which case $v$ is added to $G'$ (as an isolated vertex). At any point in the process, we use $N_{G}(S)$ to refer to the vertices not in $S$, that are connected to $S$ by edges in $G$. Whenever $S = \emptyset$ (e.g. the beginning), the algorithm initializes $S$ by querying all unqueried vertices in $G$ until the first successful query. Then, when $S \neq \emptyset$, the algorithm queries all unqueried vertices in $N_{G}(S)$ (this process possibly increases the size of $S$, in which case the algorithm repeats this process with the new $S$). Whenever the algorithm runs out of unqueried vertices to query in $N_G(S)$, it adds $S$ and all of its edges (i.e. $(v,w) \in E$ s.t $v,w \in S$) to $G'$ and resets $S$ to be the empty set. At this point, the algorithm again attempts to initialize $S$ by querying unqueried vertices in $G$, as described earlier. The algorithm finishes when there are no more unqueried vertices in $G$.

		If there is a connected component $V_i$ of size $x+1$ or higher, $|S|$ must have reached $x+1$ at some intermediate point during the above process. Consider $S$ at the moment it reaches a size of  $|S| = x+1$. Suppose the most recent vertex added to $S$ is labeled $v$. To reach this stage, we could have made at most $|S \cup N_G(S-v)|\leq \Delta (|S|-1) = x \Delta  $ queries (as each vertex in $S-v$ has at most $\Delta$ neighbors), and exactly $x+1$ of them have been successful.
		Thus the probability of forming an $V_i$ of size $x+1$ or higher:
		\begin{align*}
			P(|V_i|>x) \leq P(\text{Bin}(x\Delta, q)> x).
		\end{align*}
		where $\text{Bin}(x\Delta, q)$ is the binomial distribution with $x\Delta$ trials and success probability $q$. 
		The expected number of successes $\mu$ in this distribution is $\mu = x\Delta q$, which is less than $x$ when $q<1/\Delta$. We can thus use Chernoff's bound to bound the probability that there are more than $x$ successes. 
    Specifically, for any $\delta > 0$, $P(\text{Bin}(x\Delta, q)> (1+\delta)\mu) \leq \left(\frac{e^{-\delta}}{(1+\delta)^{(1+\delta)}})\right)^{\mu}$. In our case, $\mu = x\Delta q$ and $1+\delta = \frac{1}{\Delta q}$, so we have
		\begin{align}
			P(|V_i|>x) &\leq \left(\frac{e^{1-\frac{1}{\Delta q}}}{(\frac{1}{\Delta q})^{\frac{1}{\Delta q}}})\right)^{x\Delta q}\\
			&\leq \left(\frac{e^{\Delta q-1}}{e^{\log(\frac{1}{\Delta q})}}\right)^{x}\\
			&\leq e^{-x(1 -q\Delta  - \log (q\Delta ))}
		\end{align}


	\end{proof}
	
	\begin{corollary} \label{Corollary:Graph}
		Let $C$ be an IQP circuit containing $k$-local diagonal gates of depth $d$ on $n$ qubits, and let $\tilde{C}$ denote the noisy implementation of $C$, where each layer is interspersed with identical Pauli noise channels $\mathcal{N}_{p_X,p_Y,p_Z}$ on every qubit. 
		Denote $p = p_Z + \min(p_X,p_Y)$. Suppose we use \cref{alg:1} to sample from $P_{\tilde{C}}$.
		Let $V_1,\ldots,V_m$ be the partition of qubits into subsets in step \ref{Alg1:Step_4}. 
		There exists a constant depth threshold $d^* \leq O(p^{-1}\log(kp^{-1})$, such that when $d > d^*$, for each $i \in 1,\ldots,m$,
		\begin{align*}
			\mathbf{P}(|V_i| > x) \leq e^{-xc_{p,k}(d)}
		\end{align*}
		where $c_{p,k}(d) = 1- (k-1)d(1-2p)^d - \log((k-1)d(1-2p)^d)$ is positive and increases with $d$ when $d > d^*$
	\end{corollary} 
	\begin{proof}
		Suppose we define $G_{\tilde{C}}$ for $\tilde{C}$ similarly to how $G_{\tilde{C}_3}$ is defined for $\tilde{C}_3$ (vertices correspond to qubits and edges correspond to entangling gates). In steps \ref{Alg1:Step_1} and \ref{Alg1:Step_2}, each qubit remains in the $\ket{+}$ state independently and with probability $(1-2p)^d$, and otherwise, it is initialized to a computational basis state. Due to steps \ref{Alg1:Step_3} and \ref{Alg1:Step_4}, the corresponding vertex of each qubit in a computational basis state will have no edges in $G_{\tilde{C}_3}$. Thus, the random process which takes $G_{\tilde{C}} \to G_{\tilde{C}_3}$ exactly corresponds to the percolation process which takes $G \to G'$ in \cref{Lemma:Percolation}, where $q=(1-2p)^d$. Note that the maximum degree of $G_{\tilde{C}}$ is $\Delta \leq (k-1)d$, because each qubit is acted on by at most $d$ gates, each of which entangles with at most $k-1$ other qubits. The result follows from \cref{Lemma:Percolation} if we define $d^*$ as the depth at which $(1-2p)^{d^*} = \frac{1}{(k-1)d^*}$, i.e. the depth after which percolation occurs with high probability ($q=1/\Delta$). By \cref{Lemma:d}, $d^* = O(p^{-1}\log(kp^{-1}))$.
	\end{proof}

        \begin{lemma} \label{Lemma:d}
        For small $p> 0$, the integer value of $d > 0$ which satisfies the equation $(1-2p)^ddc = 1$ for any constant $c>0$, is $d = \Theta(p^{-1}\log(cp^{-1}))$
        \end{lemma}
        \begin{proof}
        We can simplify,
        \begin{align}
        (1-2p)^ddc &= 1\\
            \implies d\log(1-2p) e^{d\log(1-2p)} &= c^{-1}\log(1-2p)\\
            \implies d &= \frac{W_{-1}(c^{-1}\log(1-2p))}{\log(1-2p)}
        \end{align}
        where $W_{-1}$ is the lambert $W$-function, i.e. $y = W_{-1}(z)$ is defined as a solution to the equation $e^yy = z$.  It is known that $W_{-1}(-z) = \Theta(\log(z))$ for small $z$ \cite{Chatzigeorgiou_2013} \cite{Corless1996}. In our case, $-c^{-1}\log(1-2p)$ is small when $p$ is small, and therefore,
        \begin{align}
        d = \Theta\left(\frac{\log (-c^{-1}\log (1-2p))} {\log(1-2p)}\right) \approx \Theta(p^{-1}\log (cp^{-1}))
        \end{align}
        where we have used the approximation $\log(1-2p) \approx -2p$ for small $p$
        \end{proof}

	\section{Proofs of Correctness for Sampling Algorithms}
	\label{Sec:Algorithm_Proof}

	\begin{theorem}[Proof of \Cref{alg:1} Correctness] \label{Theorem:Main_appendix}
		Suppose $C$ is an IQP circuit containing $k$-local diagonal gates of depth $d$ on $n$ qubits. Let $\tilde{C}$ denote the noisy implementation of $C$, where each layer is interspersed with identical Pauli noise channels $\mathcal{N}_{p_X,p_Y,p_Z}$ on every qubit. 
		Let $p = p_Z + \min(p_X,p_Y)$. There exists a constant depth threshold $d_c \leq O(p^{-1}\log(kp^{-1}))$ such that
		when $d \geq d_c$, there exists a randomized classical algorithm that exactly samples from $P_{\tilde{C}}$ with random runtime $T$ of expected value,
		\begin{align*}
			\mathbf{E} [T] \leq O(dn^5).
		\end{align*}
	\end{theorem}
	
	\begin{proof}
		By \cref{lemma:correctness}, the algorithm samples exactly from the output distribution. Now, to bound the runtime, observe that the first 4 steps of the algorithm take runtime $O(nd)$ as it is simply processing the circuit description, which involves $O(nd)$ channels. Step \ref{Alg1:Step_5} is the costliest, and we will loosely bound its runtime in terms of the cost of simulating each subset of qubits in $\{V_1,\ldots V_m\}$ (the partition of qubits into disjoint subsets found in step \ref{Alg1:Step_4}).
		
		For each $i \in 1,\ldots,m$, the subcircuit corresponding to $V_i$ has at most $|V_i| d \leq nd$ gates and noise channels. Each Pauli noise channel can be simulated by inserting a Pauli gate into the subcircuit with a probability specified by the channel parameters. Each gate is an $O(1)$-sparse unitary matrix acting on a $2^{ |V_i|}$-length state vector which involves $O(2^{|V_i|})$ arithmetic operations (multiplication and addition). Each multiplication can be performed in $O(n^2)$ operations because the numbers can range between $1$ and $\geq 1/\sqrt{2^n}$ which requires $O(n)$ bits of precision. Therefore, the runtime of step \ref{Alg1:Step_5} is
		\begin{align} \label{eq:statevector}
			T \leq \sum_{i=1}^m cdn^32^{|V_i|}
		\end{align}        
		for some constant $c$ which arises from the state vector simulation subroutine. Now, using linearity of expectation, we compute the expected runtime $T$ as,
		\begin{align*}
			\mathbf{E} [T] &\leq \sum_{i=1}^m \sum_{x = 0}^n c d n^3   2^x  \mathbf{P}(|V_i| = x)\\
			&= \sum_{i=1}^m \sum_{x = 0}^nc d n^3 2^x  (\mathbf{P}(|V_i|> x-1)- \mathbf{P}(|V_i| > x))\\
			&= \sum_{i=1}^m \sum_{x = 0}^{n-1} cdn^3  (2^{x+1} - 2^{x})\mathbf{P}(|V_i| > x)\\
			&=  \sum_{i=1}^m \sum_{x = 0}^{n-1} cdn^3 2^{x}\mathbf{P}(|V_i| > x) \\
			&\leq  \sum_{i=1}^m \sum_{x = 0}^{n-1} cdn^3   2^{x} e^{-xc_{p,k}(d)}\\
			&\leq cdn^5  e^{-x(c_{p,k}(d)- \log 2)}
		\end{align*}
		where we have used \cref{Corollary:Graph} in the second last step, and the fact that $m \leq n$ in the last step. 
        We define $d_c$ to be the depth at which $c_{p,k}(d_c) = \log 2$.
		Then when $d \geq d_c$, $e^{-x(c_{p,k}(d)- \log 2)} \leq 1$, giving: 
		\begin{align}
			\mathbf{E} [T] \leq cdn^5
		\end{align}
		We solve $c_{p,k}(d_c) = \log 2$ to find that $d_c$ is the depth at which $(k-1)d_c(1-2p)^{d_c} \approx 0.855$. 
		The fact that $d_c = O(p^{-1}\log(kp^{-1}))$ follows from \cref{Lemma:d}.
	\end{proof}

	\subsection{The Monte Carlo Algorithm}
	
	Using standard techniques, we can turn the Las Vegas algorithm specified in \cref{alg:1} into a Monte Carlo algorithm with bounded error.
	
	\begin{algorithm}[Monte Carlo Sampler] \label{alg:Monte_Carlo_Sampler} \ \newline 
		\begin{enumerate}
			\item Perform algorithm \ref{alg:1} up till step \ref{Alg1:Step_4}, where we have the sets $\{V_1,\ldots,V_m\}$.
			\item If:
			\begin{align*}
				\max_{i}|V_i| \leq {\frac{\log(n/\epsilon)}{c_{p,k}(d)}} 
			\end{align*}
			then the algorithm continues to step \cref{Alg1:Step_5} of algorithm \ref{alg:1} and returns a bit string.
			Otherwise, it samples a bit string from the uniform distribution $U(\{0,1\}^n)$ and returns this.
		\end{enumerate}
	\end{algorithm}
	
	\noindent We now bound the error of \cref{alg:Monte_Carlo_Sampler} in the following corollary.
	
	\begin{corollary}[Monte Carlo Algorithm Performance]
		Using the same notation as \cref{Theorem:Main_appendix}, there exists a constant depth threshold $d^* \leq O(p^{-1}\log(kp^{-1}))$, such that, when $d > d^*$,$p = \Omega(1)$ and $k=O(1)$, there exists a randomized classical algorithm that samples from $\tilde{Q}_{\tilde{C}}$, such that $\|\tilde{Q}_{\tilde{C}} - P_{\tilde{C}}\|_{TVD} \leq \epsilon$ for any $\epsilon > 0$, with worst-case runtime $T \leq O(d\poly(n/\epsilon))$
	\end{corollary}
	
	\begin{proof}
		First we show that the algorithm indeed results in an $\epsilon$-approximation to the output distribution. Using a union bound, we can bound the probability that we have to output the uniform distribution as,
		\begin{align*}
			\mathbf{P}\left(\max_i |V_i| > {\frac{\log(n/\epsilon)}{c_{p,k}(d)}}\right) &\leq   \sum_{i=1}^m \mathbf{P}\left(|V_i|>\frac{\log(n/\epsilon)}{c_{p,k}(d)}\right)\\
			&\leq n e^{-\frac{\log(n/\epsilon)}{c_{p,k}(d)} c_{p,k}(d)} \\
			&\leq \epsilon 
		\end{align*}
		where we have used the fact that $m \leq n$. Let $\tilde{Q}_{\tilde{C}}(s)$ be the output distribution of \cref{alg:Monte_Carlo_Sampler}. Let $Q_{\tilde{C}}$ be the output distribution of \cref{alg:1}. Let $E$ be the event that $\max_{i}|V_i| \leq {\frac{\log(n/\epsilon)}{c_{p,k}(d)}} $ and $E'$ be its complement. We will separate $Q_{\tilde{C}}(s)$ into two conditional distributions: $Q_{\tilde{C}}(s |E)$ and $Q_{\tilde{C}}(s |E')$, where $s$ is the output string.
		We then write:
		\begin{align*}
			\|\tilde{Q}_{\tilde{C}} - P_{\tilde{C}}\|_{TVD} &= \|\tilde{Q}_{\tilde{C}} - Q_{\tilde{C}}\|_{TVD} \\
			&= \|P(E')(\tilde{Q}_{\tilde{C}}(s |E') -  Q_{\tilde{C}}(s|E') ) \\
			& + P(E)(\tilde{Q}_{\tilde{C}}(s|E) -  Q_{\tilde{C}}(s|E) )\|_{TVD}  \\
			&\leq  \epsilon\|\tilde{Q}_{\tilde{C}}(s |E') -  Q_{\tilde{C}}(s|E')  \|_{TVD} \\
			& + (1-\epsilon)\|\tilde{Q}_{\tilde{C}}(s|E) -  Q_{\tilde{C}}(s|E) \|_{TVD}  \\
			&\leq  \epsilon\|\tilde{Q}_{\tilde{C}}(s |E') - Q_{\tilde{C}}(s|E')\|_{TVD}  \\
			&\leq \epsilon
		\end{align*}
		Now, using \cref{eq:statevector}, the fact that $\max_i |V_i| \leq {\frac{\log(n/\epsilon)}{c_{p,k}(d)}} $, and the fact that $m \leq n$, we can bound the runtime as 
		\begin{align}
			T \leq cdn^4 2^{{\frac{\log(n/\epsilon)}{c_{p,k}(d)}} } = cdn^4\left(\frac{n}{\epsilon}\right)^{\frac{\log 2}{c_{p,k}(d)}}
		\end{align}
		for some constant $c$ which arises from the state vector simulation subroutine. Note that this is $O(d\poly(n/\epsilon))$ when $c_{p,k} = \Omega(1)$, which occurs when $p \in \Omega(1)$ and $k \in O(1)$. 
	\end{proof}
	
	\section{Proofs of Hardness Construction Lemmas}
 
\begin{lemma} (from \cite{fujii2016computational})
    There exists a family of IQP circuits $C$ constructed from $e^{iZ\pi/8}$ and $e^{iZZ\pi/4}$ gates of depth 5 such that it is hard to exactly sample from $P_{C,\pfail}$ for $\pfail < 0.134$, 
\end{lemma}
\begin{proof}
    This can be shown by constructing topologically protected MBQC (Measurement Based Quantum Computing) circuits that use cluster states as a resource. These circuits are hard to exactly sample from (assuming the PH doesn't collapse), even if they are exposed to a constant level of noise on each qubit in the cluster state, due to fault-tolerance techniques. In general, MBQC circuits include intermediate adaptive measurement steps, which fall out of the IQP framework. However, using the added power of post-selection, one can simulate MBQC using samples from Hadamard-basis measurement outcomes on the cluster state. This shows that it is hard to sample Hadamard-basis measurements on a cluster state, which exactly corresponds to the short-depth IQP circuit described. Moreover, this hardness exhibits robustness to noise through an error detection argument (we can post-select on error-free outcomes). For more details, refer to \cite{fujii2016computational} \footnote{The main idea of robustness to constant noise is captured in Theorem 1 and Corollary 2, while the particular circuit family and noise threshold are in section ``A Sharp CQC Boundary"}.
\end{proof}

\begin{lemma} (from \cite{Bremner_2017})
Let $C$ be an arbitrary IQP circuit constructed of depth $d$ on $n$ qubits. Then, for any parameter $r$, there is an encoded circuit $C'$ where every $k$-local gate in $C$ is encoded into a $kr$-local gate in $C'$, and a decoding algorithm $A$ such that,
    \begin{align}
        A(P_{C',q}) = P_{C, \pfail}
    \end{align}
    where $\pfail \leq (4q(1-q))^{r/2}$
\end{lemma}
\begin{proof}
    Consider the repetition code, which is given by an $n \times nr$ generator matrix $M$ which is a concatenation of $r$ $n \times n$ identity matrices in a row. For any bitstring $i\in \{0,1\}^n$, we will use $Z_i \in \{I,Z\}^n$ to denote a tensor product of $Z$ and $I$ Pauli matrices, where there is a $Z$ in the location of every $1$ in $i$ and an $I$ otherwise. \cite{Bremner_2017} showed that any circuit $C$ on $n$ qubits can be transformed into a more noise-robust circuit $C'$ on $nr$ by transforming each gate as $e^{iZ_i\pi/8} \to e^{i Z_{iM}\pi/8}$ (we use $iM$ to denote matrix multiplication between $i$ and $M$, which results in the codeword corresponding to $i$).
    
    Samples from $P_C$ can be recovered by taking a majority vote over each block of $r$ physical qubits in $P_{C'}$. The probability that a logical bit is decoded incorrectly is the same as the probability that more than r/2 bits are flipped, which is 
    \begin{align}
        \sum_{i>r/2}q^i(1-q)^{r-i} \leq 2^rq^r(1-q)^r = (4q(1-q))^{r/2}
    \end{align}
    These incorrectly decoded qubits are essentially bit-flip errors in the distribution of the logical output bits. For more details, refer to \cite{Bremner_2017} \footnote{See section ``Fault-tolerance"}.
\end{proof}

\begin{lemma} (from \cite{shepherd2010binary}
    Let $C$ be an arbitrary IQP circuit constructed using gates of the form $e^{iZ^{\otimes k}\pi/8}$ for some $k$ a multiple of $3$, and gates of the form $e^{iZ^{\otimes k'}\pi/4}$ for some $k'$ a multiple of $2$. 
    There exists an equivalent circuit $C'$, such that $P_{C'} = P_C$, where every $e^{iZ^{\otimes k}\pi/8}$ gate in $C$ is decomposed into $k^2/2$ layers of $3$-local gates in $C'$, and every $e^{iZ^{\otimes k'}\pi/8}$ gate in $C$ is decomposed into $k'$ layers of $2$-local gates in $C'$.
\end{lemma}

\begin{proof}
    In \cite{shepherd2010binary}, it was shown that gates of the form $e^{iZ^{\otimes k}\pi/8}$ can be simplified into 3-local gates as follows. We use $|i|$ to denote the hamming weight of bitstring $i$. For any bitstring $i\in \{0,1\}^n$, we use $\ketbra{1}_i \in \{I,\ketbra{1}\}^n$ to denote a tensor product of of $\ketbra{1}$ and $I$  matrices, where there is a $\ketbra{1}$ in the location of every $1$ in $i$ and an $I$ otherwise.

\begin{align}
    \frac{\pi}{8}Z^{\otimes k} &= \frac{\pi}{8}(I-2\ketbra{1})^{\otimes k}\\
    &= \frac{\pi}{8}\sum_{i \in \{0,1\}^k}  2^{|i|}\ketbra{1}_{i}\\
    &\cong \frac{\pi}{8}\sum_{i \in \{0,1\}^k:|i| \leq 3}2^{|i|} \ketbra{1}_{i}\\
    &= \frac{\pi}{8}\sum_{i \in \{0,1\}^k:|i| \leq 3} 2^{|i|} \left[\frac{I-Z}{2}\right]_{i}
\end{align}
where we have used the fact that rotation angles which are multiples of $2\pi$ are trivial.
Note that the expansion of the final equation involves only weight-1,2,3 $Z$ strings. Every possible weight-1 $Z$ string can be implemented in 1 layer of 1-local diagonal gates. 
Every possible weight-2 $Z$ string can be implemented in $k$ layers of $2$-local diagonal gates by Vizing's theorem \cite{Vizing} (which states that a complete graph can be edge-colored with at most $k$ colors). Every possible weight-3 $Z$ string can be implemented in ${k-1 \choose 2}$ layers of $3$-local gates when $k$ is a multiple of 3 by Baranyai's theorem \cite{BARANYAI1979276} (which states that every complete hypergraph, where each hyperedge is of size $r$, can be hyperedge-colored with ${k-1 \choose r-1}$ colors). Thus, the depth required to implement a gate of this form is $\leq (k-1)(k-2)/2 + k + 1 \leq k^2/2$. Similar analysis yields the fact that $e^{ik'\pi/4}$ gates can be decomposed into ${k'-1 \choose 1} + 1 = k'$ layers of $2$-local gates when $k'$ is a multiple of $2$. For more details, refer to \cite{shepherd2010binary} \footnote{See section ``Equivalence modulo $\theta$"}.
\end{proof}

\end{document}